\documentclass[preprint,aps,floatfix,nofootinbib,a4paper, superscriptaddress]{revtex4-1}
\pdfoutput=1

\usepackage{amsmath, amssymb, amsfonts, amsthm, latexsym, epsfig, mathrsfs, xcolor, bbm, slashed, braket, cancel}

\usepackage[inline]{enumitem}

\usepackage{setspace}
\usepackage[marginal, multiple]{footmisc}

\usepackage[T1]{fontenc}
\usepackage[utf8]{inputenc}
\usepackage{lmodern}

\usepackage[colorlinks, allcolors=blue!70!black, linktocpage]{hyperref}

\renewcommand\thesection{\arabic{section}}
\renewcommand\thesubsection{\arabic{subsection}}

\makeatletter
\def\p@subsection{\thesection.}
\def\p@subsubsection{\thesection.\thesubsection.}
\makeatother 

\theoremstyle{plain}

\newtheorem{lemma}{Lemma}[section]

\theoremstyle{definition}

\theoremstyle{remark}

\newcommand{\be}{\begin{equation}\begin{aligned}}
\newcommand{\ee}{\end{aligned}\end{equation}}

\newcommand{\bb}{\mathbb}

\definecolor{indigo(dye)}{rgb}{0.0, 0.25, 0.42}

\newcommand{\antiHilb}{%
\hspace{4pt} 
  \vbox{%
    \hrule height 0.5pt
    \kern0.25ex
    \hbox{%
      \kern-0.3em
      \ifmmode\Hilb\else\ensuremath{\Hilb}\fi
      \kern0em
    }
  }
}

\begin{document}
\count\footins = 800 
\setstretch{1.2}


\title{Cross-Section Continuity of Definitions of Angular Momentum}


\author{Po-Ning Chen}
\email{po-ning.chen@ucr.edu}
\affiliation{Department of Mathematics\\ University of California\\ Riverside, CA, USA.}

\author{Daniel E. Paraizo}
\email{dep5397@psu.edu}
\affiliation{Enrico Fermi Institute and Department of Physics\\ The University of Chicago\\ 5640 South Ellis Avenue, Chicago, IL 60637, USA.}
\affiliation{Institute for Gravitation and the Cosmos and Department of Physics, Pennsylvania State University, University Park, PA, 16802}

\author{Robert M. Wald}
\email{rmwa@uchicago.edu}
\affiliation{Enrico Fermi Institute and Department of Physics\\ The University of Chicago\\ 5640 South Ellis Avenue, Chicago, IL 60637, USA.}

\author{Mu-Tao Wang}
\email{mw2007@columbia.edu}
\affiliation{Department of Mathematics\\ Columbia University\\ New York, NY, USA.}

\author{Ye-Kai Wang}
\email{ykwang@math.nctu.edu.tw}
\affiliation{Department of Applied Mathematics\\ National Yang Ming Chiao Tung University\\ Hsinchu, Taiwan.}
\affiliation{National Center of Theoretical Science, Taipei City, Taiwan}

\author{Shing-Tung Yau}
\email{styau@tsinghua.edu.cn}
\affiliation{Department of Mathematics\\ Tsinghua University\\ Beijing, China.}

\begin{abstract}

We introduce a notion of ``cross-section continuity'' as a criterion for the viability of definitions of angular momentum, $J$, at null infinity: If a sequence of cross-sections, ${\mathcal C}_n$, of null infinity converges uniformly to a cross-section ${\mathcal C}$, then the angular momentum, $J_n$, on ${\mathcal C}_n$ should converge to the angular momentum, $J$, on ${\mathcal C}$. The Dray-Streubel (DS) definition of angular momentum automatically satisfies this criterion by virtue of the existence of a well defined flux associated with this definition. However, we show that the one-parameter modification of the DS definition proposed by Compere and Nichols (CN)---which encompasses numerous other alternative definitions---does not satisfy cross-section continuity. On the other hand, we prove that the Chen-Wang-Yau (CWY) definition does satisfy the cross-section continuity criterion.

\end{abstract}

\maketitle
\tableofcontents

\section{Introduction}
\label{sec:intro}

General relativity is a complete classical theory of gravity. It does not require definitions of auxiliary quantities such as mass and angular momentum in order to make physical predictions. Nevertheless, such auxiliary quantities can be extremely useful for proving general results on the behavior of systems as well as for obtaining physical insight into the nature of various phenomena. In particular, the notion of Bondi mass---together with its positivity and the positivity of its flux---provides an extremely powerful tool for characterizing and constraining possible behaviors in asymptotically flat spacetimes.

One would expect a notion of angular momentum also to be very useful for characterizing the behavior of systems in asymptotically flat spacetimes. However, a notable difficulty in obtaining a useful notion of angular momentum arises from the fact that in general relativity the asymptotic symmetry group is the Bondi-Metzner-Sachs (BMS) group. The BMS group significantly enlarges the Poincare group by including ``angle dependent translations'' known as supertranslations, which are characterized by an arbitrary function on a sphere. As will be elucidated further in the next paragraph, in a general asymptotically flat spacetime, there is no unique way of selecting a preferred Poincare subgroup of the BMS group. Consequently, there is no obvious way of distinguishing between a ``pure rotation'' and a ``rotation plus supertranslation.'' Correspondingly, there is no obvious way of distinguishing between ``angular momentum'' and ``angular momentum plus a supertranslation charge.'' 

The nature of the difficulty in selecting a preferred Poincare subgroup of asymptotic symmetries in a generic asymptotically flat spacetime is elucidated by the following considerations. During a sufficiently long era where the Bondi news vanishes, there exists a unique $4$-parameter family of cross-sections of null infinity--- which we refer to as ``electric parity good cuts''---on which the electric parity part of the shear tensor\footnote{The usual notion of ``good cuts'' requires that both the electric and magnetic parity parts of the shear vanish. We require only the vanishing of the electric parity part of the shear.} vanishes \cite{NP}. In such eras, there is a unique Poincare subgroup of the BMS group that maps these electric parity good cuts into themselves \cite{NP}. Thus, in eras where the Bondi news vanishes, one can naturally eliminate the ``supertranslation ambiguity'' of the BMS group by restricting to this Poincare subgroup. In particular, if the news vanishes sufficiently rapidly at early retarded times ($u \to - \infty$), one can naturally pick out a preferred Poincare subgroup at asymptotically early times. Similarly, if the news vanishes sufficiently rapidly at late retarded times ($u \to + \infty$), one can naturally pick out a preferred Poincare subgroup at asymptotically late times. The problem is that, generically, these early and late time subgroups will be {\em different} Poincare subgroups of the BMS group. Indeed, the {\em memory effect} is characterized by having the ``electric parity good cuts'' at asymptotically early and late times differ by a nontrivial supertranslation. The fact that a nonvanishing memory effect generically is present in asymptotically flat spacetimes provides a clear demonstration that, generically, it cannot be useful to attempt to restrict consideration to a single Poincare subgroup of asymptotic symmetries.

Given the above situation, there are two possible strategies that can be employed to attempt to define a notion of angular momentum. The first is to abandon the attempt to identify a unique notion of a ``rotation'' and a corresponding ``angular momentum'' but rather work with the entire group of BMS symmetries. One then defines notions of ``charge'' conjugate to all BMS symmetries. Of course, any individual might choose to declare a particular BMS symmetry to be a ``rotation'' and they might then refer to the corresponding BMS charge as ``angular momentum.'' However, different individuals may make different choices. In other words, this approach has the drawback that there would be an inherent supertranslation ambiguity in what BMS charge should be called ``angular momentum.'' Furthermore, as described above, at asymptotically early and late retarded times, one does have a well defined notion of a ``pure rotation'' (as opposed to a ``rotation plus supertranslation''). However, in the presence of memory, no fixed BMS charge can simultaneously correspond to this natural notion of angular momentum in both of these asymptotic regions.

A second strategy is to use data on a given cross-section to effectively determine what should be considered to be a ``pure rotation'' and ``angular momentum'' at the particular retarded time represented by that cross-section. In this way, the supertranslation ambiguity would be eliminated and---since the definition does not require use of a fixed BMS symmetry---it can be chosen to correspond to the well defined notions of ``pure rotation'' and ``angular momentum'' in eras when the Bondi news vanishes. Thus, this strategy has the potential to yield a notion of angular momentum in general relativity with properties much more analogous to the notion of angular momentum in non-gravitational physics. A potential drawback of this strategy is that there may be circumstances where there is a distinguished BMS symmetry---such as when the spacetime admits an axial Killing field---and there is no guarantee that this approach will yield a quantity that can be interpreted as the charge associated with this symmetry.

A successful implementation of the first strategy was given by Dray and Streubel (DS) \cite{DS}. As shown in \cite{WZ}, the formulas for the DS charges conjugate to the BMS symmetries can be given a strong motivation from Hamiltonian considerations. Nevertheless, possibilities remain open for modifying the DS formulas for the BMS charges. Recently, Compere and Nichols (CN) \cite{CN} have proposed a one-parameter modification of the DS definition. This modification encompasses other definitions that had been given previously \cite{LL}-\cite{CFR}. Thus, there remain questions as to whether the DS definition is the only viable definition within the context of the first strategy.

A successful implementation of the second strategy was given by Chen-Wang-Yau (CWY) \cite{CWY}. As shown in  \cite{CWWY, CKWWY}, the limit of the Chen-Wang-Yau quasilocal angular momentum at null infinity (evaluated in \cite{KWY}) is free of supertranslation ambiguity. In addition, the CWY angular momentum transforms according to classical laws with respect to ``ordinary translations'' \cite[eq.(22)]{CWWY}.

The main purpose of this paper is to subject the above definitions to the criterion of ``cross-section continuity'': If an arbitrarily small amplitude ``wiggling'' of a cross-section in retarded time can produce finite changes in the angular momentum, then that definition is not viable, since it would be telling one at least as much about the choice of cross-section as any physical properties of the spacetime. In section \ref{angmomdef}, we review the definitions of the DS, CN, and CWY angular momentum and we then introduce the notion of cross-section continuity. The DS definition can be seen to satisfy this condition by virtue of the existence of a well-defined flux. However, in section \ref{cn}, we show that the CN definition does not satisfy this condition except for the choice of parameter where it coincides with the DS definition. Thus, the CN generalization is not viable. Finally, in section \ref{cwy}, we prove that the CWY definition does satisfy the cross-section continuity requirement.


\section{The DS, CN, and CWY Definitions of Angular Momentum}
\label{angmomdef}

We will work in Bondi coordinates $(u,r,x^A)$, where $x^A$ denotes angular coordinates on the $2$-sphere. We will follow the conventions of Flanagan and Nichols \cite{FN} and refer the reader to that reference for the definitions of the shear tensor, $C_{AB}(u,x^A)$, mass aspect, $m(u, x^A)$, and angular momentum aspect, $N_A(u, x^A)$. The Bondi news tensor is defined by $N_{AB} = \partial_u C_{AB}$. We denote the round metric on the sphere as $q_{AB}$. We denote the corresponding derivative operator on the sphere as $D_A$, and we write $D^2 = D^A D_A$.

Our analysis could be performed for arbitrary BMS symmetries. However, to keep both the discussion and the formulas as simple as possible, we will restrict consideration here to the BMS symmetry 
\be
X = Y^A \frac{\partial}{\partial x^A}
\label{Xrot}
\ee
where $Y^A$ is a fixed rotational Killing field on the spheres of constant $u$. This symmetry corresponds to a ``pure rotation'' in our Bondi coordinates $(u,r,x^A)$. We will refer to the charge corresponding to this symmetry as ``angular momentum.'' However, we will be interested in evaluating the angular momentum on an arbitrary cross-section given by
\be
u = f(x^A)
\ee
where $f: S^2 \to {\bb R}$ is smooth. Note that X does not act tangentially to such a cross-section. Indeed, in the new Bondi coordinates 
\begin{eqnarray}
 u' &=& u - f(x^A) \nonumber \\
x'^A &=& x^A 
\label{bct}
\end{eqnarray}
for which the cross-section is given by $u'=0$, we have
\be
X = Y^A \frac{\partial}{\partial x'^A} - Y^A D_A f \frac{\partial}{\partial u'} \, .
\label{X2}
\ee
Thus, in the new Bondi coordinates, $X$ corresponds to the rotation $Y^A$ plus a supertranslation by the amount $-Y^A D_A f$. This illustrates the supertranslation ambiguity in defining ``rotations'' that was described in the Introduction.

As discussed in the Introduction, the DS definition does not attempt to resolve the supertranslation ambiguity and simply defines an angular momentum quantity $J^{\rm DS}$ conjugate to the BMS symmetry $X$. On the cross-section $u=0$, the formula for the DS angular momentum  is\footnote{A formula for the DS charges in Bondi coordinates is given in eq.(3.5) of \cite{FN}. A restriction to vanishing news was made in the section of \cite{FN} containing that equation but, as shown in \cite{GPS}, the charge formula is valid without this restriction.}
\be
J^{\rm DS}\vert_{u=0} =  \frac{1}{8 \pi} \int Y^A \bigg (N_A -\frac{1}{4} C_{AB}D_D C^{DB} \bigg) \, 
\ee
where all quantities in the integrand are evaluated at $u=0$ and the integral is taken over the sphere with the natural volume element associated with $q_{AB}$. (In the following, unless otherwise noted, all integrals are taken over a sphere with this volume element.)
We can evaluate $J^{\rm DS}$ on the cross-section $u = f$ by transforming to the new Bondi coordinates eq.~(\ref{bct}) and taking into account the fact that $X$ has an additional supertranslation part in these coordinates (see eq.~(\ref{X2})). On the cross-section $u=f$, the formula for DS angular momentum is 
\be
J^{\rm DS}\vert_{u=f} =  \frac{1}{8 \pi}  \int Y^A \bigg (N'_A -\frac{1}{4}Y^A C'_{AB}D_D C'^{DB} \bigg)  - \frac{1}{4 \pi}  \int m' Y^A D_A f \, 
\label{jdsf}
\ee 
where the last term is the supertranslation contribution. The quantities $C'_{AB}$, $m'$, and $N'_A$ are the shear tensor, mass aspect, and angular momentum aspect in the new Bondi coordinates eq.~(\ref{bct}) evaluated at $u'=0$ (i.e., $u=f$). Explicitly, we have (see section C.5 of \cite{CJK})
\begin{eqnarray}
    C'_{AB} &=& C_{AB}\big\vert_{u=f} - 2 D_A D_B f + q_{AB} D^2 f  \label{Ctrans} \\
    m' &=& m|_{u=f} +\frac{1}{4}N^{AB}\big\vert_{u=f}D_A D_B f +\frac{1}{2}[D_B N^{AB}]\big\vert_{u=f} D_A f  + \frac{1}{4}[\partial_u N^{AB}]\big\vert_{u=f}D_A f D_B f.
    \label{mtrans}
\end{eqnarray}
There is a similar (although considerably more complicated) formula for $N'_A$ but we will not need this formula in our analysis below. Note that it follows immediately from eq.~(\ref{Ctrans}) that the news tensor transforms simply as
\be
N'_{AB} = N_{AB}\vert_{u=f}.
\ee

The CN definition of angular momentum follows the same strategy as DS of not attempting to resolve the supertranslation ambiguity but it makes a modification of the DS formula by the addition of a term with an arbitrary free parameter $\alpha$. For the BMS symmetry $X$ given by eq.(\ref{Xrot}), the CN angular momentum at $u=0$ is given by \cite{CN}
\be
J^{\rm CN}\vert_{u=0} = J^{\rm DS}\vert_{u=0} - \frac{\alpha - 1}{32 \pi} \int Y^A C_{AB} D_C C^{BC} \, .
\ee
This coincides with the DS angular momentum when $\alpha = 1$. As Compere and Nichols have noted \cite{CN}, the choice $\alpha = 3$ corresponds to the definitions used by Landau and Lifshitz \cite{LL}, Bonga and Poisson \cite{BP}, and Damour \cite{Dam}, whereas the choice $\alpha = 0$ corresponds to the definitions used by Pasterski, Strominger, and Zhiboedov \cite{PSZ} and Compere, Fiorucci, and Ruzziconi \cite{CFR}. On the cross-section $u=f$, the CN angular momentum is given by
\be
J^{\rm CN}\vert_{u=f} = J^{\rm DS}\vert_{u=f} - \frac{\alpha - 1}{32 \pi} \int Y^A C'_{AB} D_C C'^{BC} \, .
\label{jcnf}
\ee

As described in the Introduction, the CWY definition of angular momentum resolves the supertranslation ambiguity by, in effect, determining what corresponds to a ``pure rotation'' on any cross-section from the data on that cross-section. To define CWY angular momentum on any cross-section, one does not need to select a BMS symmetry $X$ but merely a Killing field $Y^A$ on the sphere. On the cross-section $u=0$, the CWY angular momentum associated with $Y^A$ is related to the DS angular momentum by \cite{CWWY, CKWWY}
\be
J^{\rm CWY}\vert_{u=0} = J^{\rm DS}\vert_{u=0} + \frac{1}{8 \pi}  \int m Y^A D_A c\label{CWY}
\ee
where $J^{\rm DS}$ is the DS angular momentum associated with $X$, eq.(\ref{Xrot}), and $c$ is the unique solution to 
\be
 \frac{1}{2} D^2 (D^2 + 2) c = D^A D^B C_{AB} 
\label{ceq}
\ee
such that $c$ has no $\ell = 0,1$ parts. Thus, the CWY angular momentum corresponds to the DS charge conjugate to the BMS symmetry
\be
\tilde{X} = Y^A \frac{\partial}{\partial x^A} + \frac{1}{2} Y^A D_A c \frac{\partial}{\partial u}
\ee
and, from this perspective, it can be viewed as selecting $\tilde{X}$ as being the BMS symmetry representing a ``pure rotation'' at the time $u=0$. Comparing with eq.~(\ref{X2}), we see that $\tilde{X}$ is tangent to the cross-section $u = c/2$. Now, in a sufficiently long non-radiative era---where the Bondi News vanishes and thus $C_{AB}$ is independent of $u$---it follows from eq.~(\ref{Ctrans}) that the electric parity part of the shear tensor vanishes on the cross-section $u=c/2$. Thus, in a non-radiative era, $\tilde{X}$ is tangent to an ``electric parity good cut'' of null infinity and thus represents a natural choice of ``pure rotation.'' In particular, the CWY definition corresponds to the natural choices of ``pure rotation'' at both asymptotically early and late retarded times---even though these early and late time choices correspond to different BMS symmetries. The CWY angular momentum on an arbitrary cross section interpolates between these choices in a manner that depends on the local conditions at that cross-section. However, if the spacetime has an exact axial Killing field, we have shown that $\tilde{X}$ need not correspond to this Killing field when the Bondi news is nonvanishing, i.e., the CWY angular momentum need not coincide with the DS (= Komar) angular momentum in the case where $X$ corresponds to an exact Killing field in the spacetime. In this case,

\be
J^{\rm DS}\vert_{u=f}=J^{\rm DS}\vert_{u=0} \text{ for any cross-section $u=f$}.
\ee
On the other hand, 
\be
J^{\rm CWY}\vert_{u=f}=J^{\rm CWY}\vert_{u=0} \text{ only if $f$ satisfies $Y^AD_A f=0$. }
\ee

On the cross-section $u=f$, a formula for the CWY angular momentum can be obtained by transforming to the new Bondi coordinates eq.~(\ref{bct}). We obtain
\be
J^{\rm CWY}\vert_{u=f} = J^{\rm DS}\vert_{u=f} + \frac{1}{8 \pi} \int m' Y^A D_A c'  +\frac{1}{4\pi} \int m' Y^A D_A f
\label{jcwyf}
\ee
where $c'$ is defined by eq.~(\ref{ceq}) with $C_{AB}$ replaced by $C'_{AB}$. 

We now introduce our cross-section continuity condition. Let $J : {\mathscr C} \to {\bb R}$ be any map from the space, $\mathscr C$, of smooth cross-sections of null infinity into $\bb R$. Then $J$ will be said to satisfy the {\em cross-section continuity condition} at $u=0$ if for any sequence $\{ f_n \}$ of smooth functions on the sphere such that $f_n \to 0$ uniformly as $n \to \infty$ we have $J\vert_{u=f_n} \to J\vert_{u=0}$. If this condition is not satisfied, then the value of $J$ would depend on the fine details of the choice of cross-section and could not plausibly provide useful physical information about the properties of the spacetime. Our goal for the remainder of this paper is to determine whether $J^{\rm DS}$, $J^{\rm CN}$, and $J^{\rm CWY}$ satisfy the cross-section continuity condition.

Since the formulas (\ref{jdsf}), (\ref{jcnf}), and (\ref{jcwyf}) for $J^{\rm DS}$, $J^{\rm CN}$, and $J^{\rm CWY}$ on the cross-section $u=f$ depend nonlinearly on derivatives of $f$, it is not immediately obvious by inspection of these formulas whether cross-section continuity holds for any of these quantities. However, it is known that the DS formula has an associated flux \cite{WZ}, \cite{GPS}. For the BMS symmetry $X$ given by eq.(\ref{Xrot}), the flux in our original Bondi coordinates is given by \cite{FN}
\be
{\mathcal F} = - \frac{1}{32 \pi} N^{AB} \pounds_Y C_{AB}.
\ee
For any two cross-sections ${\mathcal C}_1, {\mathcal C}_2 \in {\mathscr C}$, the difference between the DS angular momenta on these cross-sections is given by
\be
J^{\rm DS}\vert_{{\mathcal C}_2} - J^{\rm DS}\vert_{{\mathcal C}_1} = \int_{\mathcal R} {\mathcal F}
\ee
where $\mathcal R$ denotes the (compact) region of null infinity bounded by ${\mathcal C}_1$ and ${\mathcal C}_2$. It follows immediately from this flux formula---taking ${\mathcal C}_2$ to be the cross-section $u=f$ and ${\mathcal C}_1$ to be the cross-section $u=0$---that cross-section continuity is satisfied by the DS definition of angular momentum. In the next section, we shall show that cross-section continuity does not hold for $J^{\rm CN}$ except in the case $\alpha = 1$, where the CN definition coincides with the DS definition. In section \ref{cwy}, we will show that cross-section continuity does hold for the CWY definition of angular momentum.


\section{Failure of Cross-Section Continuity of the CN Modification of the DS Angular Momentum}
\label{cn}

Since $J^{\rm DS}$ satisfies the cross-section continuity condition, it is clear that $J^{\rm CN}$ satisfies this condition if and only if the quantity
\be
K \equiv J^{\rm CN} - J^{\rm DS}
\ee
satisfies this condition. We have
\be
K\vert_{u=f} - K\vert_{u=0} = - \frac{\alpha - 1}{32 \pi} \int Y^A \left[ C'_{AB} D_C C'^{BC} - C_{AB} D_C C^{BC} \right].
\label{Hdiff}
\ee
Let $(\theta, \phi)$ denote the usual spherical coordinates on the sphere. We consider the Killing field $Y^A = (\partial/\partial \phi)^A$ and investigate whether the right side converges to zero for the choice of sequence 
\begin{equation}    
    f_n(\theta, \phi) = \frac{1}{n} F(\theta) \sin(n \phi)
\label{fnseq}
\end{equation}
where $F$ is a smooth function of $\theta$ that vanishes in a neighborhood of $\theta = 0$ and $\theta =\pi$. Note that any expression containing at least $p$ factors of $f_n$ and a total of at most $q$ angular derivatives, $D_A$, is $O(1/n^{p-q})$ as $n \to \infty$, so any term with more factors of $f_n$ than angular derivatives can be neglected in this limit.
Taylor expanding $C_{AB}\big\vert_{u=f_n}$ in eq.~(\ref{Ctrans}), we obtain
\begin{equation}
    C'_{AB} = C_{AB}\big\vert_{u=0} + f_n N_{AB}\big\vert_{u=0} + \frac{1}{2} f_n^2 \partial_u N_{AB}\big\vert_{u=0} + O(f_n^3) - 2 D_A D_B f_n + q_{AB} D^2 f_n  
\label{Cexp}
\end{equation}
where the $O(f_n^3)$ term arises from the Taylor expansion and does not contain angular derivatives of $f_n$. This $O(f_n^3)$ term appears on the right side of eq.~(\ref{Hdiff}) only in terms that contain more factors of $f_n$ than angular derivatives, so 
for the sequence (\ref{fnseq}), the $O(f_n^3)$ term cannot contribute to the right side of eq.~(\ref{Hdiff}) as $n \to \infty$. The term $\frac{1}{2} f_n^2 \partial_u N_{AB}\vert_{u=0}$ potentially could contribute as $n \to \infty$ via terms that are cubic in $f_n$ with a total of three $\phi$-derivatives acting on $f_n$. An example of such a term is
\begin{eqnarray}
    \frac{1}{2}Y^A f_n^2 \partial_u N_{AB} (-&2& D_C D^B D^C f_n) \approx - f_n^2 \partial_u N_{\phi \phi} (q^{\phi \phi})^2  D_\phi D_\phi D_\phi f_n \nonumber \\
    &\approx& - f_n^2 \partial_u N_{\phi \phi} \frac{1}{\sin^4 \theta} \frac{\partial^3 f_n}{\partial \phi^3} \approx F^3(\theta) \partial_u N_{\phi \phi} \frac{1}{\sin^4 \theta} \sin^2 (n \phi) \cos(n \phi) \label{quadr}
\end{eqnarray}
where $\approx$ denotes equality modulo terms that vanish for the sequence (\ref{fnseq}) as $n \to \infty$. Although the right side of eq.(\ref{quadr}) remains bounded as $n \to \infty$, its integral over a sphere vanishes as $n \to \infty$. Similarly, all other terms involving $\frac{1}{2} f_n^2 \partial_u N_{AB}\vert_{u=0}$ do not contribute as $n \to \infty$. The remaining terms on the right side of eq.~(\ref{Hdiff}) are either linear in $f_n$ or quadratic in $f_n$. For the terms that are linear in $f_n$, all derivatives acting on $f_n$ can be removed by integration-by-parts, in which case it is clear that these terms cannot contribute as $n \to \infty$. The right side of eq.(\ref{Hdiff}) is thus reduced to the quadratic terms
\begin{multline}
-\frac{\alpha - 1}{32\pi}\int Y^A \bigg [ f_nN_{AB}D_C (f_n N^{BC}) + \frac{1}{2}C_{AB}D_C(f^2_n \partial_u N^{BC}) -2f_nN_{AB}D_C D^B D^C f_n\\ + f_nN_{AB} D^B D^2f_n  +\frac{1}{2}f_n^2 \partial_u N_{AB} D_C C^{BC} - 2D_A D_B f_n D_C (f_nN^{BC}) + D^2 f_n D^B(f_nN_{AB})\\ +4D_A D_B f_n D_C D^B D^C f_n - 2D_A D_B f_n D^B D^2 f_n -2D^2 f_n D_B D_A D^B f_n +D^2 f_n D_A D^2 f_n  \bigg]. \label{fnterms}
\end{multline}
The terms appearing in (\ref{fnterms}) can be broken into 3 types: (i) terms quadratic in the news; (ii) terms linear in the news; (iii) terms that are independent of the news. The terms quadratic in the news contain at most one derivative of $f_n$ and cannot contribute to the right side of eq.~(\ref{Hdiff}) as $n \to \infty$. The terms that are independent of the news can be shown to give vanishing contribution by a calculation similar to the calculation that shows that the CN angular momentum vanishes for an arbitrary cross-section in Minkowski spacetime (see section III of \cite{EN}). Thus, we need only consider the contribution to the right side of eq.~(\ref{Hdiff}) arising from terms that are linear in the news and quadratic in $f_n$. After some cancelations arising from integration by parts, we obtain
\begin{equation}
K\vert_{u=f_n} - K\vert_{u=0} \approx  \frac{\alpha - 1}{16 \pi} \int Y^A \left[f_nN_{AB}D_CD^BD^C f_n  + D_AD_B f_n D_C (f_n N^{BC}) \right] \, .
\label{Q3}
\end{equation}
In the first term of (\ref{Q3}), we may  write $D_C D^B \approx D^B D_C$ since the difference yields a curvature term that will not contribute as $n \to \infty$. Integrating this term by parts with respect to $D^B$ after this interchange, we obtain
\begin{equation}
K\vert_{u=f_n} - K\vert_{u=0} \approx  \frac{\alpha - 1}{16 \pi} \int Y^A \left[- q_{AB} D^2 f_n  + D_AD_B f_n \right] D_C (f_n N^{BC}). 
\label{Q6}
\end{equation}
Evaluating the right side for our choices of $Y^A$ and of $f_n$ and discarding terms that do not contribute as $n \to \infty$, we obtain
\begin{equation}
K\vert_{u=f_n} - K\vert_{u=0} \approx  \frac{\alpha - 1}{16 \pi} \int N^{\theta \phi} \cos^2(n \phi) \left[F \partial_\theta F - F^2 \cot \theta \right] \, .
\label{Q4}
\end{equation}
Thus, we find
\begin{equation}
\lim_{n \to \infty} K\vert_{u=f_n} - K\vert_{u=0} = - \frac{\alpha - 1}{64 \pi} \int F^2(\theta)  \left( \frac{\partial N^{\theta \phi}}{\partial \theta} + 3 \cot \theta N^{\theta \phi} \right) \, .
\label{Q5}
\end{equation}
Since $F(\theta)$ is an arbitrary smooth function that vanishes near the poles, it is clear that the integral on right side of eq.~(\ref{Q5}) is nonvanishing in general, so the CN definition of angular momentum does not satisfy the cross-section continuity condition except in the case $\alpha = 1$, when it coincides with the DS definition.

\section{Satisfaction of Cross-Section Continuity of the CWY Angular Momentum}
\label{cwy}
\subsection{Relating the fluxes of  $J^{\rm CWY}$ and $J^{\rm DS}$}

We prove the cross-section continuity of the CWY angular momentum in the following sense: The CWY angular momentum of a smooth cross section $u=f$ approaches the CWY angular momentum of the cross section $u=0$, provided $f$ approaches $0$ in $C^0$. Here $f$ is any smooth function on $S^2$. 
Such a continuity statement precipitates the extension of the definition of CWY angular momentum to all $C^0$ cross sections. 
The proof will rely on the cross-section continuity of the DS angular momentum.

 The strategy of the continuity proof consists of the following two steps.

In step 1),  assuming $f(x^A)>0$, we relate the fluxes of  CWY angular momentum and DS angular momentum in the region between $u=0$ and $u=f$. In particular, we identify $\delta (f)$ in 
\begin{equation}\label{delta_f} J^{\rm CWY}\vert_{u=f}-J^{\rm CWY}\vert_{u=0}={J}^{\rm DS}\vert_{u=f}-{J}^{\rm DS}\vert_{u=0}+\delta (f)\end{equation} as an integral on $S^2$. 

In step 2), we apply elliptic estimates to $\delta(f)$. In particular, we show that 
\[\delta(f)\leq C||f||_{C^0}\] for a constant $C$. 
By the continuity of the DS angular momentum and the above two steps,  the proof of the cross-section continuity of the CWY angular momentum will be completed.

First, combining eq.\eqref{CWY} and eq.\eqref{jcwyf} yields
\begin{equation}\begin{split}&J^{\rm CWY}\vert_{u=f}-J^{\rm CWY}\vert_{u=0}\\
&={J}^{\rm DS}\vert_{u=f}-{J}^{\rm DS} \vert_{u=0}+\frac{1}{4\pi}\int m' Y^A D_A f+\frac{1}{8\pi}\int m' Y^A D_A c'-\frac{1}{8\pi}\int m Y^A D_A c.\end{split}\end{equation}
Therefore, we identify $\delta(f)$ as 
\begin{equation}\label{delta_f2}\delta (f)=\frac{1}{8\pi}\int m' Y^A D_A (c'+2f)-\frac{1}{8\pi}\int m Y^A D_A c.\end{equation}

We proceed to relate $c'$ and $c$. Introducing the  unique function $s$ on $S^2$ of $\ell\geq 2$ that satisfies \begin{equation}\label{s_eq}\frac{1}{2}D^2 (D^2+2)s=D^AD^B(C_{AB}\vert_{u=f}-C_{AB}\vert_{u=0}),\end{equation} eq.\eqref{Ctrans} implies $c'=s+c-2f$. Therefore,  $\delta (f)$ can be written as
\begin{equation}\label{delta_f}\delta (f)= \frac{1}{8\pi} \int  m' Y^A D_A (s + c) -\frac{1}{8\pi} \int  mY^A D_A c.\end{equation} 

Finally, we relate $m'$ to $s$. Recalling the modified mass aspect  \begin{equation}\label{mma} \widehat{m} (u, x^A)=(m-\frac{1}{4}D^AD^B C_{AB})(u, x^A)\end{equation}as in \cite{CKWWY}, eq.\eqref{mtrans} is then equivalent to 
\begin{equation}\label{relation}m'=\widehat{m}\vert_{u=f}+\frac{1}{8} D^2 (D^2+2) (s+c).\end{equation}

\subsection{Estimating $\delta (f)$}

In this subsection, we show that $\delta (f)$ is bounded by $\| f\|_{C^0}$. We first rewrite eq.\eqref{delta_f} as
\begin{align*}
8\pi\delta (f)
&= \int m' Y^A D_A s + \int (m' -m) Y^A D_A c \\ &= \delta_1 + \delta_2
\end{align*}
The following identity for a smooth function $g$ on $S^2$ is proved in \cite[Lemma 2.1]{CKWWY}:
\begin{equation}\label{integral}\int (D^2 (D^2+2) g) Y^A D_A g=0.\end{equation}

By eq.\eqref{integral}, eq.\eqref{relation}, and integration by parts,
\[\delta_1 =\int (\widehat{m}\vert_{u=f}+\frac{1}{8} D^2 (D^2+2)c ) Y^AD_A s\]
and 
\[\begin{split}\delta_2&=\int (\widehat{m}\vert_{u=f}-m+\frac{1}{8}D^2 (D^2+2)c ) Y^A D_A c+\frac{1}{4}\int (C_{AB}\vert_{u=f}-C_{AB}\vert_{u=0})D^AD^B (Y^E D_E c)\\
&=\int \left[ \int_0^f \partial_u \widehat{m} \, du \right] Y^AD_Ac+\frac{1}{4}\int \left[\int_{0}^f \partial_u C_{AB}\, du \right] D^AD^B (Y^E D_E c),  \end{split}\] where we used $m-\frac{1}{8}D^2 (D^2+2)c=\widehat{m}\vert_{u=0}$. Recalling $\partial_u \widehat{m}=-\frac{1}{8}N_{AB} N^{AB}$ from \cite[(3.3)]{CKWWY}, it is clear that $|\delta_2|$ is bounded by a multiple of $\| f\|_{C^0}$.

On the other hand, we claim that
\[ | \delta_1 | \le C_1 \| f\|_{C^0} \]  by the following Lemma.  

\begin{lemma}\label{estimate}
If $s$ satisfies $\frac{1}{2}D^2(D^2 + 2)s = D^AD^B M_{AB}$, then 
\begin{align*}
\| D s\|_{L^2} \le  \sqrt{\frac{16\pi}{3}} \max_{S^2} |M_{AB}|.
\end{align*}
\end{lemma}
\begin{proof}
Multiplying the equation by $s$ and integrating over $S^2$,
we get
\[ \int (D^2s)^2 - 2 |D s|^2 = \int 2 M_{AB} D^AD^B s \le \int 2 |M_{AB}|^2 + \int \frac{1}{2}|D_AD_B s|^2. \]
Recalling the identity on $S^2$, \begin{align}\label{identity}
\int (D^2 s)^2 = \int |D_AD_B s|^2 + |D s|^2,
\end{align} we get
\[ \int (D^2 s)^2 -3 |D s|^2 \le 16\pi (\max_{S^2} |M_{AB}|^2). \]

Let $s = \sum_{\ell=2}^\infty s_\ell$ be the spherical harmonic decomposition of $s$. We have $\int (D^2 s)^2 = \sum_{\ell= 2}^\infty \ell^2(\ell+1)^2 s_\ell^2$, $\int |D s|^2 = \int -sD^2 s = \sum_{\ell=2}^\infty \ell (\ell + 1) s_\ell^2$ and hence $\int (D^2 s)^2 \ge 6 \int |D s|^2$. As a result, we have  $\int |D s|^2 \le \frac{16\pi}{3} \max_{S^2} |M|^2$.
\end{proof}

\section*{Acknowledgements}

This material is based upon work supported by the National Science
Foundation under Grant Number DMS-2104212 (Mu-Tao\ Wang) and Grant PHY-2105878 (D.E. Paraizo and R.M. Wald), by the Simons Foundation under Grant  Number 584785 (Po-Ning\ Chen), and by Taiwan NSTC grant
109-2628-M-006-001-MY3 (Y.-K.\ Wang).

\end{document}